\def\be{\begin{equation}}
\def\ee{\end{equation}}
\def\ba{\begin{array}}
\def\ea{\end{array}}

\documentclass[priprint,showpacs,amsmath]{revtex4}
\usepackage{amssymb}
\usepackage{graphicx,epsfig,subfigure,dsfont,amssymb,amsthm,amsfonts,amsbsy,mathrsfs,amscd}
\def\qed{\leavevmode\unskip\penalty9999 \hbox{}\nobreak\hfill
     \quad\hbox{\leavevmode  \hbox to.77778em{%
               \hfil\vrule   \vbox to.675em%
               {\hrule width.6em\vfil\hrule}\vrule\hfil}}
     \par\vskip3pt}

\newtheorem{theorem}{Theorem}

\newtheorem{lemma}{Lemma}
\input amssym.def
\begin{document}
\title{Criterion for SLOCC Equivalence of Multipartite Quantum States}
\author{Tinggui Zhang$^{1}$}
\author{Ming-Jing Zhao$^{2}$}
\author{Xiaofen Huang$^{1\ast}$}
\affiliation{ $^1$School of Mathematics and Statistics, Hainan
Normal University, Haikou, 571158, China\\
$^2$Department of Mathematics, School of Science, Beijing
Information Science and Technology University, 100192, Beijing,
China\\
$\ast$Corresponding author, e-mail address:huangxf1206@163.com}

\begin{abstract}
We study the stochastic local operation and classical communication
(SLOCC) equivalence for arbitrary dimensional multipartite quantum
states. For multipartite pure states, we present a necessary and sufficient criterion
 in terms of their coefficient matrices. This condition can be used to classify some SLOCC equivalent
 quantum states with coefficient matrices having the same rank. For multipartite mixed state,
 we provide a necessary and sufficient condition by means of the realignment of matrix.
 Some detailed examples are given to identify the SLOCC equivalence of multipartite quantum states.
\end{abstract}

 \maketitle

\section{Introduction}

 Quantum entanglement is not only a prime feature in
quantum mechanics but also an important resource in quantum
information processes \cite{nils,rpmk}. It can be used in quantum
teleportation \cite{cgcr,dsfl}, superdense coding \cite{chsj,khpa},
quantum computation \cite{ddeu,pwsh,aerj,lkgr}, quantum key
distribution \cite{chgb,ngwh} and etc. Therefore, it is important to understand what kind of entanglement a
given quantum state has. One approach to classify entanglement is by means of Statistic local operations and classical communications (SLOCC)
\cite{csdj}. Entanglement in bipartite pure states has been well understood, while many questions are still open for the
mixed states and multipartite states.

It has been shown that two pure states $|\varphi\rangle$ and
$|\psi\rangle$ in ${\cal H}_1\otimes{\cal H}_2\otimes\cdots\otimes{\cal H}_K$, ${\rm dim}{\cal H}_i=n_i$, $i=1,2,\cdots,K$,
are SLOCC equivalent if and only if they can be converted into each other with
the tensor products of invertible local operators(ILOs)
\begin{equation}
|\varphi\rangle=A_1\otimes A_2 \otimes \cdots \otimes A_K|\psi\rangle.
\end{equation}
Correspondingly, two mixed states $\rho$ and $\rho^{\prime}$ belong to
the same class under SLOCC if and only if they are converted by ILOs
with nonzero determinant, that is,
\begin{equation}
\rho^{\prime}=(A_1\otimes
A_2\otimes \cdots \otimes A_K)\rho (A_1\otimes A_2\otimes \cdots
\otimes A_K)^{\dag},
\end{equation}
where $A_i$ is ILO in $GL(n_i,\mathbb{C})$ for
each $i$ \cite{wgji}. Many researches have been conducted
on entanglement classification under SLOCC since the
beginning of this century
\cite{wgji,adma,cejs,fjbh,lcyx,jscf,blhf,ocjs,lili,syjj,ggnr}.
In three-qubit system, all pure states are classified into six types
\cite{wgji}. This classification can be extended to three-qubit
mixed states \cite{adma}. Even though, it is still a very difficult
problem to find a SLOCC class of a given three-qubit mixed state
except for a few rare case. For instance, a complete SLOCC
classification for the set of the GHZ-symmetric states was reported
in Ref. \cite{cejs}. In four-qubit case, all pure states are
classified into nine SLOCC inequivalent families using group theory
\cite{fjbh}. For $n$-qubit system, Ref. \cite{lili} uses the ranks of
the coefficient matrices to study SLOCC classification for pure
state. Then Ref. \cite{syjj} generalizes Li's approach to n-qudit
pure state. Recently, Ref. \cite{ggnr} shows that almost all SLOCC
equivalent classes can be distinguished by ratios of homogeneous
SL-invariant polynomials of the same degree. Theoretically, their
technique can be applied to any number of qudits in all dimensions.
But, it is still a significant challenge to find a general scheme
that is able to completely identify the different entanglement
classes and determine the transformation matrices connecting two
equivalent states under SLOCC for multipartite mixed states. In Ref.
\cite{nmxt}, we have constructed a nontrivial set of invariants for
any multipartite mixed states under SLOCC.

In this paper we present a general scheme for the SLOCC equivalence
of arbitrary dimensional multipartite quantum pure or mixed states in
terms of matrix realignment \cite{racr,kcla}. In Sec. II, we recall
some basic results, then we give the criterion for how to judge a
block invertible matrix can be decomposed as the tensor products of invertible matrices. In Sec. III, we
give a necessary and sufficient criterion for
the SLOCC equivalence of multipartite pure states. For the
multipartite mixed states, we propose a similar criterion based on
the density matrix itself in Sec. IV. These criteria are shown to be
still operational for general states, and we also give the explicit
forms of the connecting matrix for two SLOCC equivalent states in
specific examples. At last, we give the conclusions and remarks.

\section{Tensor products decomposition for block invertible matrix}

First we introduce the definitions for realignment of matrix
\cite{racr,kcla}.

\noindent {\bf Definition 1:} For any $M\times N$ matrix $A$ with
entries $a_{ij}$, $vec(A)$ is defined by
$$
vec(A)\equiv[a_{11},\cdots,a_{M1},a_{12}\cdots,a_{M2},\cdots,a_{1N},\cdots,a_{MN}]^T,
$$
where $T$ denotes transposition.

\noindent {\bf Definition 2:} Let $Z$ be an $M\times M$ block matrix
with each block of size $N\times N$, the realigned matrix $R(Z)$ is
defined by
$
R(Z)\equiv[vec(Z_{11}),\cdots,vec(Z_{M1}),\cdots,vec(Z_{1M}),\cdots,vec(Z_{MM})]^T.
$

Based on the definitions of realignment, Ref. \cite{lljc} shows a necessary and sufficient condition for the tensor products decomposition of invertible matrices for a matrix.

\begin{lemma}\label{lemma tensor decomposition}
An $MN\times MN$ invertible matrix $A$ is
expressed as the tensor product of an $M\times M$ invertible matrix
$A_1$ and an $N\times N$ invertible matrix $A_2$, i.e, $A=A_1\otimes
A_2$ if and only if rank $R(A)=1$.
\end{lemma}

For any $N_1N_2\cdots N_K\times N_1N_2\cdots N_K$
matrix $A$, we denote $A_{i|\widehat{i}}$ the $N_i\times N_i$ block
matrix with each block of size $N_1N_2\cdots N_{i-1} N_{i+1}\cdots
N_K\times N_1N_2\cdots N_{i-1} N_{i+1}\cdots N_K$. Namely, we view
$A$ as a bipartite partitioned matrix $A_{i|\widehat{i}}$ with
partitions $H_i$ and $H_1\otimes H_2 ... H_{i-1}\otimes
H_{i+1}...H_K$. Accordingly, we have the realigned matrix
$R(A_{i|\widehat{i}})$.

\begin{theorem}\label{tensor decomposition}
Let $A$ be an $N_1N_2\cdots N_K\times
N_1N_2\cdots N_K$ invertible matrix, there exist $N_i\times N_i$
invertible matrices $a_i$, $i=1,2,\cdots, K$, such that
$A=a_1\otimes a_2 \otimes \cdots\otimes a_K$ if and only if the
rank$(R(A_{i|\widehat{i}}))=1$ for all $i$.
\end{theorem}

\begin{proof}

First, if there exist $N_i\times N_i$
 invertible matrices
$a_i$, $i=1,2,\cdots, K$, such that $A=a_1\otimes a_2 \otimes \cdots
\otimes a_K$, by viewing $A$ in bipartite partition and using Lemma
\ref{lemma tensor decomposition}, one has directly that rank$(R({A_{i|\widehat{i}}}))=1$ for all
$i$.

On the other hand, if rank$(R({A_{i|\widehat{i}}}))=1$, for any
given $i$, we prove the conclusion by induction. First, for $n=3$,
from Lemma \ref{lemma tensor decomposition}, we have $A=a_1\otimes a_{23}=a_2\otimes a_{13}$. Multiplying $a_1^{-1}$ for the first subsystem from the left, it has
$(a_1^{-1}\otimes I_{2}\otimes I_{3})A=I_1\otimes a_{23}=a_2\otimes
((a_1^{-1}\otimes I_3)a_{13})$. By tracing out the first subsystem,
we get $N_1a_{23}=a_2\otimes Tr_1((a_1^{-1}\otimes I_3)a_{13})$,
i.e, $a_{23}=a_2\otimes a_3^{\prime}$ with invertible matrix
$a_3^{\prime}=Tr_1((a_1^{-1}\otimes I_3)a_{13})/N_1$. Assume that
the conclusion is also true for $K-1$, then for $K$, from Lemma \ref{lemma tensor decomposition},
we have $A=a_1\otimes a_{\widehat{1}}=a_2\otimes
a_{\widehat{2}}=\cdots=a_K\otimes a_{\widehat{K}}$, where $a_i$ is
an $N_i\times N_i$ invertible matrix and $a_{\widehat{i}}$ is an
$N_1 N_2\cdots N_{i-1} N_{i+1}\cdots N_K\times N_1 N_2\cdots N_{i-1}
N_{i+1}\cdots N_K$ invertible matrix, $i=1,2,\cdots,K$. Hence
$(I_1\otimes\cdots \otimes I_{K-1}\otimes
a_K^{-1})A=(I_1\otimes\cdots \otimes I_{K-1}\otimes
a_K^{-1})(a_1\otimes a_{\widehat{1}})=\cdots=(I_1\otimes\cdots \otimes I_{K-1}\otimes
a_K^{-1})(a_{\widehat{K}}\otimes
a_K)$. Tracing out the last subsystem we get $a_1\otimes
Tr_K(I_2\otimes\cdots \otimes I_{N_{K-1}}\otimes
a_K^{-1})a_{\widehat{1}}))=\cdots=Tr_K((I_1\otimes\cdots \otimes
I_{K-2}\otimes a_K^{-1})\otimes(a_{K-1})=N_K a_{\widehat{K}}$. Based
on the assumption, we know $a_{\widehat{K}}$ can be written as the
tensor products of local invertible operators. Therefore, $A$ also can be
written as the tensor products of local invertible operators, which completes
the proof.

\end{proof}

\section{criterion for multipartite pure states}

First, we recall the notations of coefficient matrices of pure
state \cite{lili,syjj}. Let $\{|i_1\rangle\}_{i_1=0}^{n_1-1}$, $\{|i_2\rangle\}_{i_2=0}^{n_2-1}$, $\cdots$, $\{|i_K\rangle\}_{i_K=0}^{n_K-1}$ be orthnormal basis of $K$ Hilbert spaces ${\cal H}_1$, ${\cal H}_2$, $\cdots$, ${\cal H}_K$. For any $K$ partite pure state
$|\psi\rangle=\Sigma_{i_1,i_2,\cdots,i_K=0}^{n_1-1,n_2-1,\cdots,n_K-1}a_{i_1i_2,\cdots,i_K}|i_1i_2,\cdots,i_K\rangle$,
$\Sigma_{i_1,i_2,\cdots,i_K=0}^{n_1-1,n_2-1,\cdots,n_K-1} |a_{i_1i_2,\cdots,i_K}|^2=1$, we associate an $m\times n$ coefficient
 matrix $M(|\psi\rangle)$ to it, $m=n_1n_2\cdots n_t$, $n=n_{t+1}\cdots n_K$, $t=[\frac{K}{2}]$.

For example, for three qubit pure state
$|\psi\rangle=\sum_{i_1,i_2,i_3=0}^1
a_{i_1i_2i_3}|i_1i_2i_3\rangle$, we have the $2\times 4$ coefficient
matrices:
\begin{eqnarray*}
M(|\psi\rangle)=\left(
\begin{array}{cccc}
a_{000} & a_{001} & a_{010} & a_{011}\\
a_{100} & a_{101} & a_{110} & a_{111}\\
\end{array}
\right).
\end{eqnarray*}
For four qubit pure state $|\psi\rangle=\sum_{s_1,s_2,s_3,s_4=0}^1
a_{s_1s_2s_3s_4}|s_1s_2s_3s_4\rangle$, there is $4\times 4$
coefficient matrices, that is:\begin{eqnarray*}
M(|\psi\rangle)&=&\left(
\begin{array}{cccc}
a_{0000} & a_{0001} & a_{0010} & a_{0011}\\
a_{0100} & a_{0101} & a_{0110} & a_{0111}\\
a_{1000} & a_{1001} & a_{1010} & a_{1011}\\
a_{1100} & a_{1101} & a_{1110} & a_{1111}
\end{array}
\right).
\end{eqnarray*}

Using the rank of coefficient matrix $M(|\psi\rangle)$, Refs. \cite{lili,syjj} classified multipartite pure states into different families. If the coefficient matrices of two pure states have different ranks, then these two pure states are not SLOCC equivalent. While
the converse does not hold true, i.e. if the coefficient matrices have the same rank, then corresponding pure states are
not necessarily SLOCC equivalent. Here we answer this question further
when two states with the same rank of the coefficient matrices are equivalent
under SLOCC.

\begin{theorem}\label{th for pure}
For two $K$-partite pure states $|\phi\rangle$
and $|\psi\rangle$, they
are SLOCC equivalent if and only if for one pair of coefficient matrices $M(|\phi\rangle)$ and $M(|\psi\rangle)$, there
are $m\times m $ unitary matrices $X_1, X_2$, invertible diagonal
matrix $B_1$,and $n\times n $ unitary matrices $Y_1, Y_2$, invertible
diagonal matrix $B_2$, such that
\begin{eqnarray}\label{slocc2}
M(|\phi\rangle)=X_1B_1X_2^{\dag}M(|\psi\rangle)Y_2^{\dag}B_2Y_1,
\end{eqnarray}
and
\begin{eqnarray}\label{eq r1 in th1}
rank[R((X_1B_1X_2^{\dag})_{i|\widehat{i}})]=1
\end{eqnarray}
and
\begin{eqnarray}\label{eq r2 in th1}
rank
[R((Y_2^{\dag}B_2Y_1)_{j|\widehat{j}})]=1,
\end{eqnarray}
$i=1,2,\cdots,t$,
$j=t+1,\cdots,K$.
\end{theorem}

\begin{proof}

First, suppose $|\phi\rangle$ and
$|\psi\rangle$ are SLOCC equivalent, i.e. there exist invertible
matrices $C_1,C_2,\cdots,C_K$ such that $|\phi\rangle=(C_1\otimes
C_2\otimes\cdots\otimes C_K)|\psi\rangle$. In matrix form,
\begin{eqnarray}\label{slocc3}
M(|\phi\rangle)=(C_1\otimes
C_2\otimes\cdots\otimes
C_t)M(|\psi\rangle)(C_{t+1}\otimes \cdots\otimes
C_K)^T.
\end{eqnarray}
For invertible matrices $C_1\otimes
C_2\otimes\cdots\otimes
C_t$ and $(C_{t+1}\otimes \cdots\otimes
C_K)^T$, by the singular value
decomposition of a matrix, there exist $m\times m$ unitary
matrices $X_1$, $X_2$, invertible diagonal
matrix $B_1$, and $n\times n $ unitary matrices $Y_1, Y_2$, invertible
diagonal matrix $B_2$ such that:
\begin{eqnarray*}
C_1\otimes
C_2\otimes\cdots\otimes
C_t&=&X_1 B_1 X_2^{\dag},\\
(C_{t+1}\otimes \cdots\otimes
C_K)^T&=&Y_1 B_2 Y_2^{\dag}.
\end{eqnarray*}
Inserting these decompositions into Eq. (\ref{slocc3}), one gets easily
Eq. (\ref{slocc2}). By Lemma \ref{tensor decomposition}, we can get Eqs. (\ref{eq r1 in th1}) and (\ref{eq r2 in th1}),
$i=1,2,\cdots,t$,
$j=t+1,\cdots,K$.

On the other hand, suppose there exist one pair of coefficient
matrices $M(|\phi\rangle)$ and $M(|\psi\rangle)$ of $|\phi\rangle$ and
$|\psi\rangle$ satisfying the conditions mentioned in the Theorem. By Lemma \ref{tensor decomposition}, we know there are invertible matrices $C_1,C_2,\cdots,C_K$
such that Eq. (\ref{slocc3}) holds true. Therefore $|\phi\rangle=(C_1\otimes C_2\otimes\cdots\otimes
C_k)|\psi\rangle$, i.e. $|\phi\rangle$ and $|\psi\rangle$ are
SLOCC equivalent.

\end{proof}

Let us now take a closer look at equations in Theorem \ref{th for pure}. Eq. (\ref{slocc3}) means if two pure states are SLOCC equivalent, then their coefficient matrices have the same rank. Eqs. (\ref{eq r1 in th1}) and (\ref{eq r2 in th1}) means if two pure states are SLOCC equivalent, then their coefficient matrices
are connected by tensor products of invertible matrices. So if the coefficient matrices have the same rank, then one needs to verify Eqs. (\ref{eq r1 in th1}) and (\ref{eq r2 in th1}) to check whether two pure states are SLOCC equivalent or not.

Operationally, for two pure states $|\phi\rangle$ and $|\psi\rangle$, we first choose one kind of coefficient matrices $M(|\phi\rangle)$ and $M(|\psi\rangle)$. If $M(|\phi\rangle)$ and $M(|\psi\rangle)$ have different ranks, then $|\phi\rangle$ and $|\psi\rangle$ are not SLOCC equivalent. If $M(|\phi\rangle)$ and $M(|\psi\rangle)$ have the same rank,
then by the singular value
decomposition, there are $m\times m$ unitary
matrices $X_1$, $X_2$, diagonal
matrix $\Lambda_1$, and $n\times n $ unitary matrices $Y_1, Y_2$,
diagonal matrix $\Lambda_2$ such that:
\begin{equation}
M(|\phi\rangle)=X_1\Lambda_1 Y_1
\end{equation}
and
\begin{equation}
M(|\psi\rangle)=X_2\Lambda_2 Y_2,
\end{equation}
where
$\Lambda_1=diag(\lambda_1,\lambda_2,\cdots,\lambda_r,0,\cdots,0)$;
$\Lambda_2=diag(\mu_1,\mu_2,\cdots,\mu_r,0,\cdots,0)$,
$\lambda_i$ and $\mu_i$ are nonzero real numbers. Let
$m\times m$ invertible matrix
$B_1=diag(\sqrt{\frac{\lambda_1}{\mu_1}},\sqrt{\frac{\lambda_2}{\mu_2}},\cdots,\sqrt{\frac{\lambda_r}{\mu_r}},1,\cdots,1)$
and $n\times n$ invertible matrix
$B_2=diag(\sqrt{\frac{\lambda_1}{\mu_1}},\sqrt{\frac{\lambda_2}{\mu_2}},\cdots,\sqrt{\frac{\lambda_r}{\mu_r}},1,\cdots,1)$, then one has
$\Lambda_1=B_1\Lambda_2B_2$ and
$M(|\phi\rangle)=X_1B_1X_2^{\dag}M(|\psi\rangle)Y_2^{\dag}B_2Y_1$. Next one needs to calculate the ranks for the realignment of $X_1B_1X_2^{\dag}$ and $Y_2^{\dag}B_2Y_1$ under all partitions to see whether it is one or not.

For bipartite pure state $|\phi\rangle=\Sigma_{i_1,i_2=0}^{n_1-1,n_2-1}a_{i_1i_2}|i_1i_2\rangle$,
there is only one way to express its coefficients in matrix form, $M(|\phi\rangle)=(a_{i_1i_2})$. Therefore, two bipartite pure states $|\phi\rangle$ and $|\psi\rangle$ are SLOCC
equivalence if and only if there exist invertible matrices $C_1,C_2$ such that
\begin{eqnarray*}
M(|\phi\rangle)=C_1M(|\psi\rangle)C_{2}^T.
\end{eqnarray*}
Or equivalently, two bipartite pure states $|\phi\rangle$ and $|\psi\rangle$ are SLOCC
equivalence if and only if their coefficient matrices have the same rank.

\section{criterion for multipartite mixed states}

\begin{theorem}\label{th for mixed}

For two
multipartite mixed quantum states $\rho_1$ and $\rho_2$, they are SLOCC equivalent if and only if
there exist $N_1N_2\cdots N_K\times N_1N_2\cdots N_K$
unitary matrices $X$ and $Y$, real diagonal invertible matrix $B$, such that
\begin{eqnarray}\label{slocc6}
\rho_1=XBY^{\dag}\rho_2YBX^{\dag},
\end{eqnarray}
and
\begin{eqnarray}\label{eq r in th2}
rank(R({XBY^{\dag}})_{i|\widehat{i}})=1,
\end{eqnarray}
for
$i= 1,2,\cdots, K$.

\end{theorem}

\begin{proof}

If $\rho_1$ and $\rho_2$ are SLOCC equivalent, then there exist invertible matrices $a_1, a_2, \cdots, a_K$ such that $(a_1\otimes a_2\otimes\cdots \otimes
a_K)\rho_1(a_1\otimes a_2\otimes\cdots\otimes a_K)^{\dag}=\rho_2$. For matrix $a_1\otimes a_2\otimes\cdots \otimes
a_K$, by singular value decomposition,
there exist $N_1N_2\cdots N_K\times N_1N_2\cdots N_K$
unitary matrices $X$, $Y$, real diagonal invertible matrix $B$,
such that $a_1\otimes a_2\otimes\cdots\otimes a_K=XBY^{\dag}$.
Then $R({XBY^{\dag}}) =R(a_1\otimes a_2\cdots \otimes
a_n)$. From Lemma \ref{lemma tensor decomposition},
$rank(R({XBY^{\dag}})_{i|\widehat{i}})=1$, for
$i= 1,2,\cdots, K$.

On the other hand, if there exist $N_1N_2\cdots N_K\times N_1N_2\cdots N_K$
unitary matrices $X$ and $Y$, real diagonal invertible matrix $B$, such that Eq. (\ref{slocc6}) holds true and $rank(R({XBY^{\dag}})_{i|\widehat{i}})=1$ for $i= 1,2,\cdots, K$, then by Lemma \ref{lemma tensor decomposition}, there exist invertible matrices $a_1, a_2, \cdots, a_K$ such that
$XBY^{\dag}=a_1\otimes a_2\otimes \cdots \otimes a_n$. Inserting this equation into Eq. (\ref{slocc6}), one gets $(a_1\otimes a_2
\otimes\cdots\otimes a_n)^{\dag}\rho_1(a_1\otimes a_2 \otimes
\cdots\otimes a_n)=\rho_2$, which ends the proof.

\end{proof}

Eq. (\ref{slocc6}) means if two mixed states are SLOCC equivalent,
then they have the same rank. Eq. (\ref{eq r in th2}) means if two
mixed states are SLOCC equivalent, then they are connected by the
tensor products of invertible matrices. Now we show how to verify
Theorem \ref{th for mixed} explicitly. For two mixed states $\rho_1$
and $\rho_2$, if they have different ranks, then they are not SLOCC
equivalent. Or else, if $\rho_1$ and $\rho_2$ have the same rank,
then we first study their spectra decompositions, \be\label{eq1}
\rho_1=X\Lambda_1 X^\dag, \ \ \rho_2=Y\Lambda_2 Y^\dag, \ee where
$X=[x_1,x_2,\cdots,x_{N_1N_2\cdots N_K}],
Y=[y_1,y_2,\cdots,y_{N_1N_2\cdots N_K}]$, $\{x_i\}$ and $\{y_i\}$
are the normalized eigenvectors of states $\rho_1$ and $\rho_2$.
$\Lambda_1=diag(\lambda_1,\lambda_2,\cdots,\lambda_r,0,\cdots,0)$;
$\Lambda_2=diag(\mu_1,\mu_2,\cdots,\mu_r,0,\cdots,0)$, $\lambda_i$
and $\mu_i$ are nonzero real numbers. For diagonal matrices
$\Lambda_1$ and $\Lambda_2$, there exists $N_1N_2\cdots N_K\times
N_1N_2\cdots N_K$ invertible matrix
\begin{eqnarray}\label{slocc1}B=diag(\sqrt{\frac{\lambda_1}{\mu_1}},
\sqrt{\frac{\lambda_2}{\mu_2}},\cdots,\sqrt{\frac{\lambda_r}{\mu_r}},s,\cdots,t)
\end{eqnarray}
such that
\begin{eqnarray*}
\Lambda_1=B\Lambda_2B,
\end{eqnarray*}
where $s,\cdots,t$ are arbitrary nonzero numbers. Therefore, there exist $N_1N_2\cdots N_K\times N_1N_2\cdots N_K$
unitary matrices $X$ and $Y$, real diagonal invertible matrix $B$, such that
Eq. (\ref{slocc6}) holds true.
Next we need to verify the rank of realignment of $XBY^{\dag}$ to see whether $\rho_1$ and $\rho_2$ are SLOCC equivalent or not.

Example 1. First, we consider two-qubit Bell-diagonal states in two-qubit system
\cite{rmhh,orud}:
$$\rho_1=\sum_{i=1}^4\lambda_i|\psi_i\rangle\langle\psi_i|,\ \ \ \lambda_i\geq 0, \ \ \ \sum_{i=1}^4\lambda_i=1, \ \ \ i=1,2,3,4;$$
$$\rho_2=\sum_{i=1}^4\mu_i|\psi_i\rangle\langle\psi_i|,\ \ \ \mu_i\geq 0, \ \ \ \sum_{i=1}^4\mu_i=1, \ \ \ i=1,2,3,4;$$
with $|\psi_1\rangle=\frac{1}{\sqrt{2}}(|00\rangle+|11\rangle)$,
$|\psi_2\rangle=\frac{1}{\sqrt{2}}(|00\rangle-|11\rangle)$,
$|\psi_3\rangle=\frac{1}{\sqrt{2}}(|01\rangle+|10\rangle)$,
$|\psi_4\rangle=\frac{1}{\sqrt{2}}(|01\rangle+|10\rangle)$.
By spectra decomposition, we have $X=Y=\left(\begin{array}{cccc}
1 & -1 & 0 & 0 \\
0 & 0 & -1 & 1 \\
0 & 0 & 1 & 1 \\
1 & 1 & 0 & 0
\end{array}\right);$
$\Lambda_1=diag(2\lambda_1,2\lambda_2,2\lambda_4,2(1-\lambda_1-\lambda_2-\lambda_4))$;
$\Lambda_2=diag(2\mu_1,2\mu_2,2\mu_4,2(1-\mu_1-\mu_2-\mu_4))$. For
simplicity, we consider only the non-degenerate case, which means $\Lambda_1$ and $\Lambda_2$ are nonsingular. Let
$B=diag(\sqrt{\frac{\lambda_1}{\mu_1}},\sqrt{\frac{\lambda_2}{\mu_2}},\sqrt{\frac{\lambda_4}
{\mu_4}},\sqrt{\frac{\lambda_1+\lambda_2+\lambda_4-1}{\mu_1+\mu_2+\mu_4-1}})$, then $\rho_1$ and $\rho_2$ satisfy Eq. (\ref{slocc6}).
Next we need to study the rank of realignment matrix
$XBY^{\dag}=B$. We find if
$$\sqrt{\frac{\lambda_1}{\mu_1}}:\sqrt{\frac{\lambda_4}
{\mu_4}}=\sqrt{\frac{\lambda_2}{\mu_2}}:\sqrt{\frac{\lambda_1+\lambda_2+\lambda_4-1}{\mu_1+\mu_2+\mu_4-1}},$$
then $rank(R(XBY^\dagger))=1$. In this case, $\rho_1$ and $\rho_2$ are SLOCC equivalent.

 Example 2. Now we consider two mixed states in $2\otimes 2 \otimes 2$
system,
$$
\rho_1=\frac{1}{K}\left(\begin{array}{cccccccc}
1 & 0 & 0 & 0 & 0 & 0 & 0 & 1\\
0 & a & 0 & 0 & 0 & 0 & 0 & 0\\
0 & 0 & b & 0 & 0 & 0 & 0 & 0\\
0 & 0 & 0 & c & 0 & 0 & 0 & 0\\
0 & 0 & 0 & 0 & \frac{1}{c} & 0 & 0 & 0\\
0 & 0 & 0 & 0 & 0 & \frac{1}{b} & 0 & 0\\
0 & 0 & 0 & 0 & 0 & 0 & \frac{1}{a} & 0\\
1 & 0 & 0 & 0 & 0 & 0 & 0 & 1
\end{array}\right),
$$
$$
\rho_2=\frac{1}{M}\left(\begin{array}{cccccccc}
1 & 0 & 0 & 0 & 0 & 0 & 0 & 1\\
0 & \alpha & 0 & 0 & 0 & 0 & 0 & 0\\
0 & 0 & \beta & 0 & 0 & 0 & 0 & 0\\
0 & 0 & 0 & \gamma & 0 & 0 & 0 & 0\\
0 & 0 & 0 & 0 & \frac{1}{\gamma} & 0 & 0 & 0\\
0 & 0 & 0 & 0 & 0 & \frac{1}{\beta} & 0 & 0\\
0 & 0 & 0 & 0 & 0 & 0 & \frac{1}{\alpha} & 0\\
1 & 0 & 0 & 0 & 0 & 0 & 0 & 1
\end{array}\right),
$$
where the normalization factors
$K=2+a+b+c+\frac{1}{a}+\frac{1}{b}+\frac{1}{c}$.
$M=2+\alpha+\beta+\gamma+\frac{1}{\alpha}+\frac{1}{\beta}+\frac{1}{\gamma}$. First we study the spectra decompositions of $\rho_1$ and $\rho_2$.
Here as in Eq. (\ref{eq1}), $\Lambda_1=
\frac{1}{K} diag (\frac{1}{c},\frac{1}{b},\frac{1}{a},2,a,b,c,0)$ and
$\Lambda_2= \frac{1}{M} diag
(\frac{1}{\gamma},\frac{1}{\beta},\frac{1}{\alpha},2,\alpha,\beta,\gamma,0)$.
To simplify the problem, suppose $a$, $b$, $c$, $\alpha$, $\beta$, $\gamma$ take different values unequal to 0, 1, $\frac{1}{2}$, 2. Then we can
easily get
$$X=Y=\left(\begin{array}{cccccccc}
0 & 0 & 0 & \frac{1}{\sqrt{2}} & 0 & 0 & 0 & -\frac{1}{\sqrt{2}}\\
0 & 0 & 0 & 0 & 1 & 0 & 0 & 0\\
0 & 0 & 0 & 0 & 0 & 1 & 0 & 0\\
0 & 0 & 0 & 0 & 0 & 0 & 1 & 0\\
1 & 0 & 0 & 0 & 0 & 0 & 0 & 0\\
0 & 1 & 0 & 0 & 0 & 0 & 0 & 0\\
0 & 0 & 1 & 0 & 0 & 0 & 0 & 0\\
0 & 0 & 0 & \frac{1}{\sqrt{2}} & 0 & 0 & 0 & \frac{1}{\sqrt{2}}
\end{array}\right).$$
Let
$B=diag(\sqrt{\frac{M\gamma}{Kc}},\sqrt{\frac{M\beta}{Kb}},\sqrt{\frac{M\alpha}{Ka}},
\sqrt{\frac{M}{K}},\sqrt{\frac{Ma}{K\alpha}},\sqrt{\frac{Mb}{K\beta}},\sqrt{\frac{Mc}{K\gamma}},C)$
with $C$ an arbitrary nonzero number. Then $XBY^{\dag}=B$. Now we calculate the rank of the realignment of $XBY^{\dag}$.
If the coefficients of $\rho_1$ and $\rho_2$ satisfies the following two condition,\\
$(1)
\sqrt{\frac{\gamma}{c}}:\sqrt{\frac{a}{\alpha}}=\sqrt{\frac{\beta}{b}}:\sqrt{\frac{b}{\beta}}=\sqrt{\frac{\alpha}{a}}:\sqrt{\frac{c}{\gamma}}=1:C,$\\

$(2)
\sqrt{\frac{\gamma}{c}}:\sqrt{\frac{\beta}{b}}=\sqrt{\frac{\alpha}{a}}:1\\
$
then $rank(R({XBY^{\dag}})_{i|\widehat{i}})=1$ for $i= 1,2,3$. In this case, $\rho_1$ and $\rho_2$ are SLOCC equivalent.
For instance, when $\sqrt{\frac{\alpha}{a}}=\sqrt{2};
\sqrt{\frac{\beta}{b}}=2;\sqrt{\frac{\gamma}{c}}=2\sqrt{2}$, one chooses $C=\frac{1}{4}$. Then such two mixed states are SLOCC equivalent.

Example 3. Let us consider another pair of mixed states in $2\otimes 2
\otimes 2$ system,
$$
\rho_1=\frac{1}{K}\left(\begin{array}{cccccccc}
1 & 0 & 0 & 0 & 0 & 0 & 0 & \frac{1}{2}\\
0 & a & 0 & 0 & 0 & 0 & 0 & 0\\
0 & 0 & b & 0 & 0 & 0 & 0 & 0\\
0 & 0 & 0 & c & 0 & 0 & 0 & 0\\
0 & 0 & 0 & 0 & \frac{1}{c} & 0 & 0 & 0\\
0 & 0 & 0 & 0 & 0 & \frac{1}{b} & 0 & 0\\
0 & 0 & 0 & 0 & 0 & 0 & \frac{1}{a} & 0\\
\frac{1}{2} & 0 & 0 & 0 & 0 & 0 & 0 & 1
\end{array}\right),
$$$$\rho_2=\frac{1}{2K}\left(\begin{array}{cccccccc}
1+b & 0 & 1-b & 0 & 0 & -1/2 & 0 & 1/2\\
0 & a+c & 0 & a-c & 0 & 0 & 0 & 0\\
1-b & 0 & 1+b & 0 & 0 & -1/2 & 0 & 1/2\\
0 & a-c & 0 & a+c & 0 & 0 & 0 & 0\\
0 & 0 & 0 & 0 & \frac{1}{c}+\frac{1}{a} & 0 & -\frac{1}{a}+\frac{1}{c} & 0\\
-\frac{1}{2} & 0 & -\frac{1}{2} & 0 & 0 & \frac{1}{b}+1 & 0 & -1+\frac{1}{b}\\
0 & 0 & 0 & 0 & -\frac{1}{a}+\frac{1}{c} & 0 & \frac{1}{c}+\frac{1}{a} & 0\\
\frac{1}{2} & 0 & \frac{1}{2} & 0 & 0 & -1+\frac{1}{b} & 0 &
1+\frac{1}{b}
\end{array}\right),
$$
where the normalization factor
$K=\frac{3}{2}+a+b+c+\frac{1}{a}+\frac{1}{b}+\frac{1}{c}$. $\rho_1$
and $\rho_2$ have the same eigenvalues, $\Lambda_1=\Lambda_2=
\frac{1}{K} diag
(\frac{1}{c},\frac{1}{b},\frac{1}{a},\frac{3}{2},a,b,c,\frac{1}{2})$.
Now we consider the case with different $a$, $b$, and $c$ unequal to 0, 1, $\frac{2}{3}$,
$\frac{3}{2}$, 2, $\frac{1}{2}$, which implies that $\rho_1$ and $\rho_2$ are not
degenerated. In such case,
$$
X=\left(\begin{array}{cccccccc}
0 & 0 & 0 & \frac{1}{\sqrt{2}} & 0 & 0 & 0 & -\frac{1}{\sqrt{2}}\\
0 & 0 & 0 & 0 & 1 & 0 & 0 & 0\\
0 & 0 & 0 & 0 & 0 & 1 & 0 & 0\\
0 & 0 & 0 & 0 & 0 & 0 & 1 & 0\\
1 & 0 & 0 & 0 & 0 & 0 & 0 & 0\\
0 & 1 & 0 & 0 & 0 & 0 & 0 & 0\\
0 & 0 & 1 & 0 & 0 & 0 & 0 & 0\\
0 & 0 & 0 & \frac{1}{\sqrt{2}} & 0 & 0 & 0 & \frac{1}{\sqrt{2}}
\end{array}\right),$$
$$Y^{\dag}=\left(\begin{array}{cccccccc}
0 & 0 & 0 & 0 & \frac{1}{\sqrt{2}} & 0 & \frac{1}{\sqrt{2}} & 0\\
0 & 0 & 0 & 0 & 0 & \frac{1}{\sqrt{2}} & 0 & \frac{1}{\sqrt{2}}\\
0 & 0 & 0 & 0 & -\frac{1}{\sqrt{2}} & 0 & \frac{1}{\sqrt{2}} & 0\\
\frac{1}{2} & 0 & \frac{1}{2} & 0 & 0 & -\frac{1}{2} & 0 & \frac{1}{2}\\
0 & \frac{1}{\sqrt{2}} & 0 & \frac{1}{\sqrt{2}} & 0 & 0 & 0 &  0\\
-\frac{1}{\sqrt{2}} & 0 & \frac{1}{\sqrt{2}} & 0 & 0 & 0 & 0 & 0\\
0 & -\frac{1}{\sqrt{2}} & 0 & \frac{1}{\sqrt{2}} & 0 & 0 & 0 &  0\\
-\frac{1}{2} & 0 & -\frac{1}{2} & 0 & 0 & -\frac{1}{2} & 0 &
\frac{1}{2}
\end{array}\right).
$$
Let $B$ be the identity matrix.
Then $$XBY^{\dag}=\frac{1}{\sqrt{2}}\left(\begin{array}{cccccccc}
1 & 0 & 1 & 0 & 0 & 0 & 0 & 0\\
0 & 1 & 0 & 1 & 0 & 0 & 0 & 0\\
-1 & 0 & 1 & 0 & 0 & 0 & 0 & 0\\
0 & -1 & 0 & 1 & 0 & 0 & 0 & 0\\
0 & 0 & 0 & 0 & 1 & 0 & 1 &  0\\
0 & 0 & 0 & 0 & 0 & 1 & 0 &  1\\
0 & 0 & 0 & 0 & -1 & 0 & 1 &  0\\
0 & 0 & 0 & 0 & 0 & -1 & 0 &  1
\end{array}\right).$$ It is easy to verify that $rank(R(XBY^{\dag})_{1|23})=1$. Furthermore
$XBY=\left(\begin{array}{cc}
1 & 0 \\
0 & 1
\end{array}\right)\otimes\frac{1}{\sqrt{2}}\left(\begin{array}{cc}
1 & 1 \\
-1 & 1
\end{array}\right)\otimes\left(\begin{array}{cc}
1 & 0 \\
0 & 1
\end{array}\right)$.
Hence all nondegenerated mixed states $\rho_1$ and $\rho_2$ are SLOCC equivalent.

Now, we give one example for two quantum states non SLOCC
equivalence. In fact, there are too many examples for two quantum
states non SLOCC equivalence.

Example 4. Suppose
$|\psi\rangle_1=\frac{1}{\sqrt{2}}(|001\rangle)+|010\rangle$,
$|\psi\rangle_2=\frac{1}{\sqrt{2}}(|101\rangle)+|011\rangle$.

On one hand, the coefficient matrices of these two pure states have
different ranks, by Theorem 2, we can easily to determine that they are non SLOCC
equivalence. On the other hand, we can also check their non SLOCC equivalence by Theorem 3. Since these are pure states and their density matrices is rank one, therefore their density matrices 
have only one nonzero eigenvalue $1$. In this case, we
can choose $B$ as identity matrix, the $X$ and $Y$ can easily
respectively obtained. One has
$$XBY^{\dag}=\frac{1}{\sqrt{2}}\left(\begin{array}{cccccccc}
1 & 0 & 0 & 0 & 0 & 0 & 0 & 0\\
0 & 0 & 0 & 1 & 0 & 0 & 0 & 0\\
0 & 0 & 0 & 0 & 0 & 1 & 0 & 0\\
0 & 1 & 0 & 0 & 0 & 0 & 0 & 0\\
0 & 0 & 1 & 0 & 0 & 0 & 0 &  0\\
0 & 0 & 0 & 0 & 1 & 0 & 0 &  0\\
0 & 0 & 0 & 0 & 0 & 0 & 1 &  0\\
0 & 0 & 0 & 0 & 0 & 0 & 0 &  1
\end{array}\right).$$
It is easy to verify that $rank(R(XBY)_{1|23})\neq 1$. By Theorem 3,
they are non SLOCC equivalence.

\section{conclusions and remarks}

We have studied the SLOCC equivalence for arbitrary dimensional
multipartite quantum states. Utilizing coefficient matrix and
realignment, we present necessary and sufficient criteria for
multipartite pure states and mixed states respectively. These
conditions can be used to classify some SLOCC equivalent quantum
states having the same rank. Some detailed examples are given to
identify the SLOCC equivalence or non SLOCC equivalence. However, our methods have to recognize its disadvantage in determining the SLOCC equivalence for degenerate state. The reason is that the normalized eigenvectors of degenerate states can not be determined up to some unitary matrix. Thus the choose of unitary matrices X and Y in Eq.(10)  can not be determined up to some unknown unitary matrices, which takes infinite possibility. Therefore, to check Eq.(10) becomes terribly difficult since one should check all possible choices.

\bigskip
Acknowledgments:  We are very thankful to an anonymous referee for
the helpful comments. This work is supported by the NSF of China
under Grant No. 11401032 and No. 11501153; the NSF of Hainan
Province under Grant Nos.(20161006, 20151010, 114006); the
Scientific Research Foundation for Colleges of Hainan Province under
Grant No.Hnky2015-18; Scientific Research Foundation for the
Returned Overseas Chinese Scholars, State Education Ministry.

\end{document}